%% file: main.tex
\documentclass{article}

\usepackage[accepted]{icml2017} 
\usepackage{times}

\usepackage{thm-restate}

\usepackage{amsmath}
\usepackage{amsthm}
\usepackage{amssymb}
\usepackage{paralist}
\usepackage{multicol}
\usepackage{amsfonts}
\usepackage{varwidth}
\usepackage{enumitem}

\usepackage[utf8x]{inputenc}

\usepackage{algorithm}% http://ctan.org/pkg/algorithms
\usepackage[noend]{algpseudocode}% http://ctan.org/pkg/algorithmicx

\usepackage{tikz}
\usetikzlibrary{fit,arrows,calc,positioning,shapes.geometric}
\usepackage{tikz-qtree}
\usepackage{standalone}
\usepackage{lipsum}

\usepackage{mathtools}

\newtheorem{theorem}{Theorem}

\newtheorem{observation}{Observation}

\newtheorem{definition}{Definition}

\input{def}

\newcommand{\xleft}{\mathopen{}\mathclose\bgroup\left}
\newcommand{\xright}{\aftergroup\egroup\right}

\newcommand{\eps}{\epsilon}
\newcommand{\rmp}{RM$^+$}

\icmltitlerunning{Regret Minimization in Behaviorally-Constrained Zero-Sum Games}

\begin{document}

%\maketitle

\twocolumn[
\icmltitle{Regret Minimization in Behaviorally-Constrained Zero-Sum Games}

\begin{icmlauthorlist}
\icmlauthor{Gabriele Farina}{cmu}
\icmlauthor{Christian Kroer}{cmu}
\icmlauthor{Tuomas Sandholm}{cmu}
\end{icmlauthorlist}
\icmlaffiliation{cmu}{Carnegie Mellon University, Pittsburgh PA 15213 USA}
\icmlcorrespondingauthor{Gabriele Farina}{gfarina@cs.cmu.edu}
\icmlcorrespondingauthor{Christian Kroer}{ckroer@cs.cmu.edu}
\icmlcorrespondingauthor{Tuomas Sandholm}{sandholm@cs.cmu.edu}

\icmlkeywords{boring formatting information, machine learning, ICML}

\vskip 0.3in
]
\printAffiliationsAndNotice{}

\input{abstract}

\input{introduction}
\input{related_work}
\input{preliminaries}
\input{rm_generalized}
\input{constraints}
\input{perturbed_nfg}
\input{perturbed_efg}
\input{experiments}

\input{conclusion}

\section*{Acknowledgments}
This material is based on work supported by NSF grant IIS-1617590 and ARO award W911NF-17-1-0082.
\bibliographystyle{named}
\bibliography{dairefs}

\clearpage\newpage
\input{appendix}

\end{document}

%% file: def.tex
\def\e{\ensuremath{{\xi}}}

\newcommand{\cI}{\ensuremath{\mathcal{I}}}
\newcommand{\cX}{\ensuremath{\mathcal{X}}}

\newcommand{\egt}{EGT}

\newcommand{\be}{\begin{eqnarray}}
\newcommand{\ee}[1]{\label{#1}\end{eqnarray}}

\newcommand{\ese}{\end{eqnarray*}}
\newcommand{\bse}{\begin{eqnarray*}}
\def\beq{\begin{equation}}
\def\eeq{\end{equation}}

\def\fnote#1{\footnote}

\def\*{{{\LARGE\bf $^*$}}}

\def\cX{{\cal X}}

\def\log{\mathop{{\rm log}}}

%% file: abstract.tex
\begin{abstract}
  No-regret learning has emerged as a powerful tool for solving extensive-form
  games. This was facilitated by the counterfactual-regret minimization (CFR) 
  framework, which relies on the instantiation of regret
  minimizers for simplexes at each information set of the game. 
%
%  We extend this framework to
%  behavioral constraints on the player strategies---by developing
%  a variant of the regret-matching$^+$ algorithm that works with additional
%  linear constraints on the simplex. Regret-matching$^+$ is the premier regret minimizer
%  used in practical large-scale extensive-form game solving.
%  
  We use an
  instantiation of the CFR framework to develop algorithms for solving behaviorally-constrained (and, as a special case, \emph{perturbed} in the Selten sense) extensive-form games, which allows us to compute approximate Nash equilibrium
  refinements.
  % Nash equilibrium has become a popular solution concept for practical design of agents in multiagent settings.
  Nash equilibrium refinements are motivated by a major deficiency in Nash equilibrium: it provides virtually no guarantees on how it will play in parts of the game tree that are reached with zero probability. Refinements can mend this issue, but have not been adopted in practice, mostly due to a lack of scalable algorithms.
  % In this paper we investigate the development of regret minimization algorithms that converge to an approximate equilibrium refinement. In particular, we extend the popular regret matching algorithm to general convex polytopes, and use this result to specialize the algorithm to perturbed simplexes. We use these results to construct a perturbed variant of the counterfactual-regret minimization algorithm, and show that it converges to an approximate equilibrium refinement.
  We show that, compared to standard algorithms, our method finds solutions that have substantially better refinement properties, while enjoying a convergence rate that is comparable to that of state-of-the-art algorithms for Nash equilibrium computation both in theory and practice.
\end{abstract}

%% file: introduction.tex
\section{Introduction}

No-regret learning algorithms have become a powerful tool for solving
large-scale zero-sum extensive-form games
(EFGs)~\cite{Bowling15:Heads-up,Brown15:Hierarchical}. This has largely been
facilitated by the \emph{counterfactual-regret minimization} (CFR) algorithm~\cite{Zinkevich07:Regret} and its newer
variants~\cite{Lanctot09:Monte,Sandholm10:State,Bowling15:Heads-up,Brown15:Regret-Based,Brown17:Dynamic,Brown17:Reduced}.
This framework works by defining a notion of regret local to an information set,
and instantiating a standard regret minimizer at each information
set in order to minimize local regret. \citet{Zinkevich07:Regret} prove that
this scheme of local regret minimization leads to a Nash equilibrium in
two-player zero-sum extensive-form games of perfect recall. The framework works
with any regret-minimizing algorithm, but in practice variants of the
\emph{regret matching} algorithm have been
dominant~\cite{Hart00:Simple,Bowling15:Heads-up,Brown15:Regret-Based,Brown15:Hierarchical}.
We investigate the extension of regret-matching$^+$
(\rmp)~\citep{Tammelin15:Solving}, an even faster regret-matching algorithm,
to more general regret-minimization problems over (finitely-generated) convex
polytopes.
 %(as opposed to just simplexes in the standard framework).
We use these results to instantiate {\rmp} for linearly constrained
simplexes, which in turn allows us to model and solve behaviorally-constrained
EFGs~(which are EFGs with additional linear constraints on the simplexes at
each information set). An important special case of this framework is
behaviorally-perturbed EFGs, which can be used to compute Nash equilibrium
refinements.\footnote{The idea of certain kinds of behavioral perturbations to CFR has been suggested by~\citet{Neller13:Introduction}. They suggest that at every information set, with small probability $\epsilon$, a player will make a random move.
  However, they provide no results of what the refinement consequences are (i.e., what kind of refinement this would lead to), and it
  is unclear whether the proposed method actually leads to refinements. In contrast, we establish a connection to (approximate) EFPEs in this paper. Furthermore, Neller and Lanctot cite~\citet{Miltersen10:Computing} which is about quasi-perfect equilibrium (where a player assumes she will not make errors in the future), while the $\epsilon$ modeling of Neller and Lanctot makes $\epsilon$ errors at future information sets as well.}

% , and combine this with CFR in order to solve perturbed extensive-form games. This allows us to develop the first regret-minimization algorithm for computing approximate equilibrium refinements in EFGs.

Nash equilibrium refinements are motivated by major deficiencies in the
Nash equilibrium solution concept: Nash equilibria provide no guarantees on
performance in information sets that are reached with probability zero in equilibrium, beyond not giving
up more utility than the value of the game. Thus, if an opponent makes a 
mistake, Nash equilibrium is not guaranteed to capitalize on it, but may instead
give back up to all of that utility~\citep{Miltersen10:Computing}. This is especially
relevant when Nash equilibria are used as a solution concept for playing against
imperfect opponents. Equilibrium refinements ameliorate this issue by
introducing further constraints on behavior in information sets that are reached
with probability zero. We will be interested in equilibrium concepts that
achieve this through the notion of \emph{perturbations} or
\emph{trembling hands}~\citep{Selten75:Reexamination}. At each decision-point, a
player is assumed to tremble with some small probability, and a Nash equilibrium
is then computed in this perturbed game. A refinement is then a limit point of
the sequence of Nash equilibria achieved as the probability of trembles is taken
to zero. In \emph{quasi-perfect equilibria}, players take into account only the
trembles of their opponents~\cite{VanDamme84:Relation}, whereas in an
\emph{extensive-form perfect equilibrium} (EFPE), players take into account
mistakes made both by themselves and opponents~\citep{Selten75:Reexamination}.

We compare our algorithm for perturbed EFGs to state-of-the-art large-scale zero-sum EFG-solving algorithms: the standard CFR$^+$ algorithm~\citep{Tammelin15:Solving} and the excessive gap technique (EGT)~\citep{Nesterov05:Excessive} instantiated with a state-of-the art \emph{smoothing function}~\citep{Nesterov05:Smooth,Hoda10:Smoothing,Kroer15:Faster,Kroer17:Theoretical}. We find that our perturbed variant of CFR$^+$ converges (in the perturbed game) at the same rate as those algorithms converge while ours leads to orders of magnitude more refined strategies. Our algorithm also converges at the same rate in the unperturbed game, almost until the point where the imposed behavioral constraints necessarily prevent further convergence.

%	\begin{itemize}
%		\item The idea of modeling a fixed error bound for the player is also found in Lanctot's tutorial. But they didn't know their model would actually provide an EFPE.
%		\item We could borrow their analysis on the symbolic value of eps being too much of a burden.
%		\item More broadly, we are taking a step in the direction of robust strategies (two previous papers by Bowling et al.)
%	\end{itemize}
%	
%		In this paper, we provide regret-based algorithms to compute approximate Nash equilibria of large-scale two-player zero-sum games. We start from the case of normal-form games, and show how to apply the analysis in the more general setting of extensive-form games. We adopt a modern formalism, already found in~\citet{Abernethy11:Blackwell} and \citet{Telgarsky11:Blackwell}, converting the notion of randomized strategies in normal-form games to a deterministic one inside a convex polytope of actions. Finally, we briefly discuss the relationship between perturbed games and perfect equilibria, a solution concept introduced by \citet{Selten75:Reexamination}.
%
%		BOTTOM LINE: We basically get to run CFR (which is super respected for these kinds of things) with no extra overhead, so (MAYBE if the experiments are promising) basically there is no reason not to use the perturbed methods.	

%% file: related_work.tex
\section{Related work}

No-regret algorithms have a long history in EFG solving.
\citet{Gordon06:No-Regret} developed the \emph{Lagrangian Hedging} algorithm,
which can be used to find a Nash equilibrium in EFGs. However, it suffers from a
drawback: it requires projection onto the strategy space at each iteration. \citet{Zinkevich07:Regret}
developed CFR, which avoids projection while maintaining the
same convergence rate. It has since been extended in a
number of ways. \citet{Lanctot09:Monte} showed how to incorporate sampling in
CFR. \citet{Brown15:Regret-Based} showed how to achieve greater pruning in CFR, thereby reducing the iteration costs. CFR$^+$ is a state-of-the-art
variant of CFR~\citep{Tammelin15:Solving}, which has
vastly superior practical performance compared to standard CFR, though it is
not known to be stronger from a theoretical perspective.
\citet{Gordon08:No-Regret} shows how no-regret algorithms can also be utilized
for computing extensive-form correlated equilibria in EFGs.

Polynomial-time algorithms have been proposed for computing certain equilibrium
refinements in two-player zero-sum perfect-recall EFGs.
\citet{Miltersen10:Computing} develop a linear program (LP) for computing
quasi-perfect equilibria by choosing a sufficiently small perturbation to
realization plans. \citet{Farina17:Extensive_Form} develop a similar approach
for EFPE computation, but rely on perturbations to behavioral strategies of
players. These approaches rely on solving modified variants of the sequence-form
LP for computing Nash equilibrium~\citep{Stengel96:Efficient} in EFGs. These algorithms are of theoretical interest only and do not work in practice. They require rational numbers of precision $n \log n$ bits, where $n$ is the number of sequences in the game. Another issue is that LP algorithms do not scale to large EFGs even when just finding Nash equilibria, and in practice CFR-based or EGT-based approaches are used to achieve scalability.
\citet{Kroer17:Smoothing} recently showed how smoothing functions for
first-order methods such as EGT can be extended to games with perturbations.

\citet{Johanson07:Computing} consider robust strategies that arise from assuming
that the opponent will randomize between playing a Nash equilibrium and a
strategy within some model of opponent behavior. \citet{Johanson09:Data} consider
a similar model-biased Nash equilibrium approach on games where an independent
model is used at each information set. \citet{Ganzfried11:Game} develop an
opponent modeling approach that adds opponent-modeled constraints across
information sets. Our approach provides a principled framework for solving
model-biased games that use general constraints on per-information set
behavioral strategies. Constraints across information sets currently require
 the much less scalable LP approach.

%% file: preliminaries.tex
\section{Preliminaries}
  We briefly introduce several of the basic concepts we use in the rest of the paper. We denote by $\mathbb{R}_+$ and $\mathbb{R}_-$ the set of non-negative and non-positive reals, respectively.
\subsection{Normal-Form Games}
\begin{definition}\label{def:std nfg}
  A \emph{two-player zero-sum normal-form game} (for the rest of the paper, simply \emph{normal-form game} or \emph{NFG}) is a tuple $({A}_1, {A}_2, u)$ where $A_1$ represents the finite set of actions that player 1 can play, $A_2$ represents the finite set of actions that player 2 can play, and $u : A_1\times A_2 \to \mathbb{R}$ is the payoff function for player 1, mapping the pair of actions $(a_1, a_2)$ of the players into the payoff for player 1. The corresponding payoff for player 2 is given by $-u(a_1,a_2)$.
\end{definition}

Usually, the payoff function $u$ is given as a matrix $U$, called the \emph{payoff matrix} of the game. The rows of $U$ represent the actions $\{a_{1,1}, \dots, a_{1,n}\} = A_1$ of player 1, while the columns of $U$ represent the actions $\{a_{2,1},\dots,a_{2,m}\} = A_2$ of player 2. At the intersection of the $i$-th row and the $j$-th column is the payoff for the action pair $(a_{1,i}, a_{2,j})$.

\begin{definition}
	A mixed strategy $\pi$ for player $i\in\{1,2\}$ is a probability mass function over the set $A_i$.
\end{definition}

When players play according to mixed strategies $\pi_1$ and $\pi_2$ respectively, the expected payoff is given by
\begin{equation}\label{eq:expected utility nfg}
  \mathbb{E}_{\pi_1,\pi_2}(u) = \sum_{a_1\in A_1}\sum_{a_2\in A_2} \pi_1(a_1)\pi_2(a_2) u(a_1, a_2).
\end{equation}

	\subsection{Generalized Normal-Form Games}
\citet{Telgarsky11:Blackwell} and~\citet{Abernethy11:Blackwell} propose a generalization of the concept of normal-form games, which conveniently allows us to remove all expectation operators, making the notation lighter and more legible. In this generalization, players select deterministic strategies from a convex compact set. For a normal-form game, this set is the space of all mixed strategies.
\begin{definition}
A \emph{two-player zero-sum generalized normal-form game} $\Gamma=(\mathcal{X},\mathcal{Y}, u)$ is a tuple defined by a pair of convex and compact action spaces $\mathcal{X} \subseteq \mathbb{R}^n$, $\mathcal{Y} \subseteq \mathbb{R}^m$, one for each player, as well as a \emph{biaffine} utility function $u : \mathcal{X}\times \mathcal{Y} \to \mathbb{R}$. The utility function $u(x,y)$ maps the pair of actions $(x, y) \in \mathcal{X}\times \mathcal{Y}$ of the players into the payoff for player 1, while the corresponding payoff for player 2 is given by $-u(x,y)$.
\end{definition}

\begin{observation}
  \label{obs:nfg to gnfg}
Any normal-form game can be mapped to an instance of a generalized normal-form game. Given $\Gamma=(A_1,A_2, u)$, where $|A_1|=n$ and $|A_2|=m$, the set of all mixed strategies for player 1 forms the $n$-dimensional simplex $\mathcal{X}=\Delta_{n}$, while the set of all mixed strategies for player 2 forms the $m$-dimensional simplex $\mathcal{Y}=\Delta_{m}$. Let $U$ be the payoff matrix associated with $\Gamma$. Using Equation~\ref{eq:expected utility nfg}, we conclude that $\Gamma$ is equivalent to the generalized two-player zero-sum normal-form game $\Gamma^* = (\mathcal{X},\mathcal{Y}, u^*)$, where $
u^*(x,y) = x^\top U y$ for all $(x,y)\in \mathcal{X}\times\mathcal{Y}$.
\end{observation}

\subsection{Extensive-Form Games}
\begin{definition}
	A \emph{two-player zero-sum extensive-form game with imperfect information and perfect recall} $\Gamma$ is a tuple $(H, Z, A, P, f_c, \mathcal{I}_1, \mathcal{I}_2, u)$ composed of:
	\begin{itemize}[nolistsep,itemsep=1mm]
	\item $H$: a finite set of possible sequences (or histories) of actions, such that the empty sequence $\varnothing \in H$, and every prefix $z$ of $h$ in $H$ is also in $H$.
	\item $Z\subseteq H$: the set of terminal histories, i.e. those sequences that are not a proper prefix of any sequence.
	\item $A$: a function mapping $h \in H\setminus Z$ to the set of available actions at non-terminal history $h$.
	\item $P$: the player function, mapping each non-terminal history $h \in H\setminus Z$ to $\{1, 2, c\}$, representing the player who takes action after $h$. If $P(h)=c$, the player is chance.
	\item $f_c$: a function assigning to each $h\in H\setminus Z$ such that $P(h) = c$ a probability mass function over $A(h)$.
	\item $\mathcal{I}_i$, for $i\in\{1,2\}$: partition of $\{h\in H: P(h)=i\}$ with the property that $A(h)=A(h')$ for each $h,h'$ in the same set of the partition. For notational convenience, we will write $A(I)$ to mean $A(h)$ for any of the $h\in I$, where $I\in \mathcal{I}_i$. ${\cal I}_i$ is the information partition of player $i$, while the sets in ${\cal I}_i$ are called the information sets of player $i$.
	\item $u$: utility function mapping $z\in Z$ to the utility (a real number) gained by player 1 when the history is reached. The corresponding utility for player 2 is given by $-u(z)$.
	\end{itemize}
	We further assume that all players can recall their previous actions and the corresponding information sets.
\end{definition}
In the rest of the paper, we will use the more relaxed term \emph{extensive-form game}, or EFG, to mean a two-player zero-sum extensive-form game with imperfect information and perfect recall.

\begin{observation}\label{obs:efg to gnfg}
Extensive-form games can be represented as generalized NFGs, for example, via the normal form representation or  sequence form representation~\cite{Romanovskii62:Reduction,Koller96:Efficient,Stengel96:Efficient}.
\end{observation}

  \subsection{Regret and Regret Minimization}
  Suppose there exists an iterative algorithm which, at each step $t = 1,\dots, T$, computes a strategy $x_t \in \mathcal{X}$ for player 1, and plays a (generalized) normal-form game $(\mathcal{X},\mathcal{Y}, u)$ against player 2 using such strategy. Let $y_t$ be the strategy used by player 2 at step $t$. The average external regret of player 1 up to step $T$ against action $\hat x \in \mathcal{X}$ is  
  \[
      \vspace{-1mm}
    \bar R_1^T(\hat x) = \frac{1}{T}\sum_{t=1}^T u(\hat x, y_t) - u(x_t, y_t).
      \vspace{-1mm}
  \]
  % The analogous definition for player 2 is symmetrical: their average external regret up to step $T$ against action $\hat y\in\mathcal{Y}$ is defined as
  % \[
  %   \bar R_2^T(\hat y) = \frac{1}{T}\sum_{t=1}^T - u(x_t, \hat y) + u(x_t, y_t).
  % \]
  The case for player 2 is symmetrical. A \emph{regret-minimizing scheme} is a
  function that assigns, for each sequence of past actions
  $x_1,y_1,\ldots,x_{t-1},y_{t-1}$, an action $x_t$ such that
  $\limsup_{T\rightarrow \infty}\max_{\hat{x}\in\mathcal{X}}\bar{R}_1^T(\hat{x})\leq 0$.

  Regret-matching~(RM)~\cite{Hart00:Simple} is a regret-minimizing scheme for normal-form games, based on Blackwell's approachability theorem~\cite{Blackwell56:Analog}.
	The following theorem, a proof of which is given by~\citet{Cesa06:Prediction}, characterizes the convergence rate of RM.
	\begin{theorem}
	  Given a normal-form game $(A_1, A_2, u)$, the maximum average external regret for player 1 at iteration $T$, when player 1 plays according to the regret-matching algorithm, is 
	  \[
      \vspace{-1mm}
		\max_{\hat x} \bar R_1^T(\hat x) \le \gamma\frac{\sqrt{|A_1|}}{\sqrt{T}},
      \vspace{-1mm}
	  \]
	  where $\gamma \doteq \max_{x,y} u(x,y) - \min_{x,y} u(x,y)$.
	\end{theorem}

	Regret-matching$^+$ is an extension of RM, and converges significantly faster in practice.  \citeauthor{Tammelin15:Solving} (2015; Lemma~2) proved that the convergence rate of RM$^+$ is the same as that of RM above.

\subsection{Nash Equilibria and Refinements}
    We now review the needed solution concepts from game theory. We mostly focus on generalized normal-form games, allowing a compact presentation of concepts pertaining both normal-form games and extensive-form games.

      \begin{definition}[Approximate best response]
        Given a generalized normal-form game $(\mathcal{X}, \mathcal{Y}, u)$ and a strategy $y\in\mathcal{Y}$, we say that $x\in\mathcal{X}$ is an $\eps$-best response to $y$ for player 1 if
        \(
          u(x,y) + \epsilon \ge u(\hat{x}, y)
        \)
        for all $\hat x \in \mathcal{X}$. Symmetrically, given $x\in\mathcal{X}$, we say that $y\in\mathcal{Y}$ is an $\eps$-best response to $x$ for player 2 if
        \(
          -u(x, y) + \epsilon \ge -u(x, \hat y)
        \)
        for all $\hat y \in \mathcal{Y}$.
      \end{definition}

	  \begin{definition}[Approximate Nash equilibrium]\label{def:approximate nash}
	    Given a generalized normal-form game $(\mathcal{X}, \mathcal{Y}, u)$, the strategy pair $(x, y) \in \mathcal{X}\times\mathcal{Y}$ is a $\epsilon$-Nash equilibrium for the game
	    if $x$ is an $\eps$-best response to $y$ for player 1, and $y$ is an $\eps$-best response to $x$ for player 2.
	  \end{definition}
	
	 \begin{definition}[Nash equilibrium]
	   Given a generalized normal-form game $(\mathcal{X}, \mathcal{Y}, u)$, a \emph{Nash equilibrium} for the game is a $0$-Nash equilibrium.
	 \end{definition}
	
	 There exists a well-known relationship between regret and approximate Nash
	  equilibria (Definition~\ref{def:approximate nash}), as summarized in the next theorem.
	
	  \begin{theorem}\label{thm:regret nash}
	    In a zero-sum game, if the average external regrets of the players up to step $T$ are such that
	    \[
      \vspace{-1mm}
	      \bar R^T_1(\hat{x}) \le \epsilon_1, \qquad
	      \bar R^T_2(\hat{y}) \le \epsilon_2
      \vspace{-1mm}
	    \]
	    for all actions $\hat{x}\in\mathcal{X}, \hat{y}\in\mathcal{Y}$,
	    then the strategy pair
	    \[
      \vspace{-1mm}
	      (\bar{x}_T, \bar{y}_T) \doteq \xleft(\frac{1}{T}\sum_{i=1}^T x_i, \frac{1}{T}\sum_{i=1}^T y_i\xright)\in\mathcal{X}\times\mathcal{Y}
      \vspace{-1mm}
	    \]
	    is an $(\epsilon_1 + \epsilon_2)$-Nash equilibrium.
	  \end{theorem}
	
	  Theorem~\ref{thm:regret nash} basically says that if there exists an iterative algorithm able to progressively propose strategies so that the maximum average external regret go to zero, then recovering a Nash equilibrium is straightforward, and just a matter of averaging the individual strategies proposed.

We now turn to the class of \emph{perturbed games}~\citep{Selten75:Reexamination}. Intuitively, a perturbation restricts the set of
playable strategies by enforcing a lower bound on the probability of
playing each action. We recall the definition and some of the properties of game
perturbations, starting from the normal-form case. We focus on player 1, but
remark that the same definitions hold symmetrically for player 2 as well.
	
\begin{definition}\label{def:pert nfg new}
  Let $\Gamma=(A_1, A_2, u)$ be an NFG and let $\Gamma^*=(\Delta_{|A_1|},
  \Delta_{|A_2|}, u^*)$ be its generalized NFG representation (see
  Observation~\ref{obs:nfg to gnfg}). A perturbation is a function $p: A_1 \cup
  A_2\to \mathbb{R}_+$ such that $\sum_{a\in A_1} p(a) < 1$ and $\sum_{a \in
    A_2} p(a) < 1$. The corresponding \emph{perturbed NFG} $\Gamma_p$ is the
  generalized NFG where each action $a$ must be played with probability at least
  $p(a)$. Formally, $\Gamma_p=(\tilde{\mathcal{X}}_p,\tilde{\mathcal{Y}}_p,u^*)$ where $\tilde{\mathcal{X}}_p=\left\{ x\in
    \Delta^{|A_1|}: x_a \geq p(a) \forall a\in A_1 \right\}$. $\tilde{\mathcal{Y}}_p$ is defined
  analogously.
\end{definition}
	
	 % The set of all available mixed strategies for player 1 is therefore as described in the following lemma.
	
	%  \begin{proposition}\label{prop:pertubed set}
	%    Any perturbed normal-form game $\Gamma_p$ is equivalent to a generalized normal-form game $(\tilde{\mathcal{X}}_p, \tilde{\mathcal{Y}}_p, u^*)$, where the \emph{perturbed action space} $\tilde{\mathcal{X}}_p$ of player 1 is the convex set
	%       \[
	%  	     	  \tilde{\mathcal{X}}_p \doteq\xleft\{\sum_{i=1}^n \lambda_i e_i \ \left|\ \begin{array}{l}\circled{\normalfont 1}\hspace{2mm} \displaystyle\lambda_i \ge p(a_i)\ \forall i\in \{1, \dots, n\}\\[2mm]
	%  	     	    \circled{\normalfont 2}\hspace{2mm}\displaystyle\sum_{i=1}^n \lambda_i = 1\end{array}
	%  	     	  \right.\xright\}.
	%  	       \]
	% \end{proposition}

    % We postpone a more thorough analysis of the properties of $\tilde{\mathcal{X}}_p$ to Section~\ref{sec:nfg}, and focus on the case of extensive-form games. In this setting, a perturbation for player 1 assigns a lower-bound on each action playable by the player. More precisely:
In the case of extensive-form games, a perturbation for player 1 assigns a lower-bound on each action playable by the player. More precisely:
    \begin{definition}
      Let $\Gamma=(H, Z, A, P, f_c, \mathcal{I}_1, \mathcal{I}_2, u)$ be an extensive-form game. A perturbation is a function $p$ mapping each pair $(I, a)$ where $I\in \mathcal{I}_1\cup\mathcal{I}_2$ and $a \in A(I)$ to a non-negative real, such that
      \[
      \vspace{-1mm}
        \sum_{a \in A(I)} p(I, a) < 1 \qquad\forall\ I\in \mathcal{I}_1\cup \cI_2.
      \vspace{-1mm}
      \]

      The corresponding \emph{perturbed EFG} $\Gamma_p$ is the analogous game
      where each action $a$ at each information set $I$ has to be played with
      probability at least $p(I,a)$.
    \end{definition}

    Perturbations play an important role in equilibrium refinement, as they form
    the basis for the concept of \emph{equilibrium
      perfection}~\citep{Selten75:Reexamination}. In this paper we only focus on
    the case of \emph{extensive-form perfect equilibria (EFPEs)}.
	\begin{definition}\label{def:efpe}
		A strategy pair $(x, y) \in\mathcal{X}\times\mathcal{Y}$ is an EFPE of $\Gamma$ if it is a limit point of a sequence $\{(x_p, y_p)\}_{p \rightarrow 0}$ where $(x_p, y_p)$ is a Nash equilibrium of the perturbed game $\Gamma_p$.
	\end{definition}
	
	Intuitively, an EFPE is an equilibrium refinement that takes into account an imperfect ability to deterministically commit to a single action. 
%\citet{Farina17:Extensive_Form} proved that an EFPE can be recovered in polynomial time in the size of the game given a $\eps$-EFPE, where $\eps$ is a small enough constant that can be computed in polynomial time in the size of the game. 

%% file: rm_generalized.tex
\section{Generalized Normal-Form Games over Finitely-Generated Convex Polytopes}\label{sec:convex polytopes}
In this section, we show how to adapt a regret-minimization algorithm to handle
generalized normal-form games played on finitely-generated convex polytopes.
The key insight is that when the action space is a finitely-generated convex
polytope, the generalized game can be cast back as a normal-form game, i.e. a
generalized normal-form game played over simplexes, and ``solved'' by a
regret-minimization algorithm; subsequently, the solution for the normal-form
game gets mapped back into the polytope. This is achieved by constructing new
simplex action spaces for the players, where each point in a simplex denotes
a convex combination of weights on the vertices of that players'
finitely-generated convex polytopal action space.
  
  \begin{theorem}\label{thm:generalized hannan}
    Let $\Gamma = (\mathcal{X}, \mathcal{Y}, u)$ be a generalized normal-form game played on the finitely-generated convex polytopes $\mathcal{X}$ and $\mathcal{Y}$. There exists a regret-minimizing scheme for player 1 in $\Gamma$.
  \end{theorem}
  \begin{proof}  
  Let $\{b_1, \dots, b_n\}$ be a convex basis for $\mathcal{X}$, and $\{c_1, \dots, c_m\}$ be a convex basis for $\mathcal{Y}$; also, let $B= (b_1 \mid\cdots\mid b_n)$ and $C=(c_1\mid\cdots\mid c_m)$ be the basis matrices for $\mathcal{X}$ and $\mathcal{Y}$, respectively. We construct a generalized normal-form game $\Gamma^* = (\Delta_n, \Delta_m, u^*)$, where
  \[
    u^*(x, y) \doteq u\xleft(Bx, Cy\xright)
  \]
  for all $x \in \Delta_n, y \in\Delta_m$. Of course $Bx \in \mathcal{X}, Cy \in \mathcal{Y}$ for all $x$ and $y$, so the definition is valid. Let $f^*$ be any of regret-minimizing schemes for normal-form games (e.g., RM or RM$^+$). We construct a regret-minimizing scheme $f$ for $\Gamma$ such that, at each iteration $t = 1, 2, \dots$, 
  \[
    f(x_1, y_1, \dots, x_{t-1}, y_{t-1}) = B f^*(x^*_1,  y^*_1, \dots, x^*_{t-1}, y^*_{t-1}),
  \]
  where $x^*,y^*$ denotes the coordinates of $x,y$ with respect to the basis of
  $\mathcal{X}, \mathcal{Y}$, respectively; note that this definition is
  well-defined since the coordinates are guaranteed to belong to $\Delta_n,\Delta_m$\footnote{Passage to coordinates might not be unique. In this case, any coordinate vector will do, as long as the choice is deterministic.}. The regret induced by this scheme is
  \begin{align*}
    \bar R^T_{1}(\hat x) &= \frac{1}{T}\sum^T_{t=1} u(\hat x, y_t) - u(x_t, y_t)\\
                         &= \frac{1}{T}\sum^T_{t=1} u^*(\hat{x}^*, y^*_t) - u^*(x^*_t, y^*_t).
  \end{align*}
  Notice that the last expression is exactly the average regret for player 1 up to iteration $T$ against action $\hat{x}^*$ in $\Gamma^*$. Since $f^*$ is a regret-minimizing scheme, the average regret against any action converges to zero, meaning that $\limsup_{T\to\infty} \bar{R}^T_1(\hat x) \le 0$ for each $\hat x$, i.e. $\limsup_{T\to\infty} \max_{\hat x\in\mathcal{X}}\bar{R}^T_1(\hat x) \le 0$. This proves that $f$ is a regret-minimizing scheme for $\Gamma$, concluding the proof.
  \end{proof}

	Another way to think about the construction above is that at each iteration, we compute the regret for not playing each of the ``strategies'' forming the vertices of the polytope, and updating the next strategy by taking a convex combination of the vertices, in a way proportional to the regret against them.

  Algorithm~\ref{algo:gnfg} represents an instantiation of the construction given in the proof of Theorem~\ref{thm:generalized hannan}, where the regret-minimizing scheme for the normal-form game was chosen to be RM$^+$. A careful analysis of the construction also reveals that the convergence bound for RM$^+$ carries over, as expressed by Theorem~\ref{thm:max regret polytope}. At time $t$, RM$^+$ projects the cumulative regret $r_{t-1}$ onto the non-negative orthant $\mathbb{R}^n_+$; the projection is equal to the vector $[r_{t-1}]^+$, where $[a]^+_i \doteq \max\{0, a_i\}$.
	
	\begin{algorithm}[h!t]
		\caption{RM$^+$ algorithm for generalized normal-form games played over finitely-generated convex polytopes.}
		\label{algo:gnfg}
		\begin{algorithmic}[1]
		\Procedure{regret-matching$^+$}{$\Gamma$}
		    \Statex $\triangleright$ \textcolor{darkgray}{$\Gamma = (\mathcal{X},\mathcal{Y}, u)$, and $B$ is a fixed convex basis for $\mathcal{X}$}
		    \Statex $\triangleright$ \textcolor{darkgray}{\textbf{note}: this reflects the point of view of player 1}
%			\Statex $\triangleright$ \textcolor{darkgray}{$c_i$ is the $i$-th column of $B_p$.}
		    \State $r_0 \gets (0, \dots, 0)^\top \in \mathbb{R}^n$
		    \State $\bar{x} \gets (0, \dots, 0)^\top \in \mathbb{R}^n$
			\For{$t = 1, 2, 3, \dots$}
			  \If{$r_{t-1} \in \mathbb{R}_-^n$}
			    \State $x_{t} \gets $ any action $\in \tilde{\mathcal{X}}_p$
			  \Else{}
			    \State $\Lambda_{t-1} \gets \displaystyle\sum_{i=1}^n\, [r_{t-1}]^+_i$
			    \State $x_t \gets \displaystyle B\frac{[r_{t-1}]^+}{\Lambda_{t-1}}$
			  \EndIf{}
			  \State play action $x_t$
			  \State observe $y_t \in \mathcal{Y}$ played by opponent\vspace{1mm}
			  \State $\displaystyle r_{t} \gets \left[r_{t-1} + \begin{pmatrix}
			  				  u(b_1, y_t) - u(x_t, y_t)\\\vdots\\
			  				  u(b_n, y_t) - u(x_t, y_t)
			  			  \end{pmatrix}\right]^+$\vspace{1mm}
			  \State $\displaystyle \bar{x} \gets \frac{t-1}{t}\bar x + \frac{1}{t}x_t$\vspace{1mm}
			\EndFor{}
			\Statex $\triangleright$ \textcolor{darkgray}{$\bar x$ contains the average strategy for player 1}
		\EndProcedure
		\end{algorithmic}
	\end{algorithm}  

	\begin{theorem}\label{thm:max regret polytope}
	  Given a generalized normal-form game $(\mathcal{X}, \mathcal{Y}, u)$ with
    finitely generated $\mathcal{X}$ and $\mathcal{Y}$, the maximum average external regret for player 1 at iteration $T$, when player 1 plays according to Algorithm~\ref{algo:gnfg}, is bounded by
	  \[
      \vspace{-1mm}
		\max_{\hat x \in \mathcal{X}} \bar R_1^T(\hat x) \le \gamma\frac{\sqrt{|\mathcal{X}|}}{\sqrt{T}}
      \vspace{-1mm}
	  \]
	  where $\gamma \doteq \max_{x, y} u(x, y) - \min_{x,y} u(x,y)$, and $|\mathcal{X}|$ denotes the number of vertices of $\mathcal{X}$.
	\end{theorem}

%% file: constraints.tex
\section{Behavioral Constraints and Perturbations}\label{sec:behavioral constraints}
% It is well-known that an extensive-form game can be cast as an instance of a
% generalized normal-form game $(\mathcal{X},\mathcal{Y},u)$, by means of the
% \emph{sequence-form representation}, which we could then add behavioral
% constraints to. However, while the action polytopes are defined by a polynomial
% number of constraints, the number of vertices is exponential (and so is the size
% of any convex basis for them); this implies that we cannot use
% Algorithm~\ref{algo:gnfg} unless we are willing to pay a cost exponential in the
% size of the game, both in terms of memory and time.
Behavioral constraint are linear constraints on the simplexes at each
information set. In order to obtain a regret minimizer for a
behaviorally-constrained EFG, we could try to cast the game as a generalized NFG
by means of the normal-form or sequence form representation (see
Observation~\ref{obs:efg to gnfg}). However, the number of vertices of this
representation is exponential, and therefore it does not work well with
Theorem~\ref{thm:max regret polytope}. Counterfactual Regret~(CFR,
\citeauthor{Zinkevich07:Regret}~\citeyear{Zinkevich07:Regret}) solves this
problem, by defining a regret-minimizing scheme that runs in polynomial time in
the size of the game. Intuitively, CFR minimizes a variant of instantaneous
regret, called \emph{immediate counterfactual regret}, at each information set
separately, and later combines the strategies computed at each information set.
It requires simplex regret minimizers for each information set. If we have a
finite number of them, each information set can be modeled as a
finitely-generated convex polytope. We can then use Theorem~\ref{thm:max regret
  polytope} to get regret minimizers for each information set. Perturbations can
be handled as a special case.

% We now specialize the results of the previous section to linearly constrained
% and perturbed simplexes. The main motivation for this is that it will allow us
% to model behavioral constraints, which , and behavioral perturbations on EFGs.

% Theorem~\ref{thm:max regret polytope} shows how to handle generalized
% normal-form games with finitely-generated polytopes. A simplex with additional
% linear constraints is a special case of this, and thus we can apply
% Algorithm~\ref{algo:gnfg} as a regret minimizer for each player. This gives us a
% regret minimizer for behaviorally-constrained information sets in EFGs.

Theorem~\ref{thm:regret cfr} below shows that CFR$^+$ instantiated with such
regret minimizers for each behaviorally-constrained information set converges to
an equilibrium of the constrained EFG. For this approach to be practical, we
need the set of vertices for each information set to be of manageable size, as
reflected in the dependence on $\max_{I\in\mathcal{I}}\sqrt{|\mathcal{Q}^I|}$ in
Theorem~\ref{thm:regret cfr}, where $|\mathcal{Q}^I|$ is the number of vertices
in the behaviorally-constrained simplex at information set $I$.

\begin{restatable}{theorem}{cfrconv}
  \label{thm:regret cfr} Let $(H, Z, A, P, f_c, \mathcal{I}_1, \mathcal{I}_2,
u)$ be an extensive-form game; let $\mathcal{Q}^I \subseteq \Delta_{|A(I)|}$
represent the behaviorally-constrained strategy space at information set $I$,
for all $I\in\mathcal{I}_1 \cup \mathcal{I}_2$. The maximum average external
regret for player 1 in the constrained game at iteration $T$, when player 1
plays according to CFR$^+$, is bounded by
        \[ 
      \vspace{-1mm}
        \bar R^T_1 \le \gamma
|\mathcal{I}_1|\frac{\sqrt{\max_{I\in\mathcal{I}_1} |\mathcal{Q}^I|}}{\sqrt{T}},
      \vspace{-1mm}
        \] where $\gamma \doteq \max_{x, y} u(x, y) - \min_{x,y} u(x,y)$.
\end{restatable}

%% file: perturbed_nfg.tex
	\section{Perturbed Normal-Form Games}\label{sec:nfg}
	Section~\ref{sec:convex polytopes} established that, in general, the problem of finding an approximate Nash equilibrium for player 1 in the generalized normal-form game $\Gamma = (\mathcal{X}, \mathcal{Y}, u)$, where $\mathcal{X}$ is a convex polytope generated by $n$ vectors, can be solved via regret-matching.
      
    We now specialize this result for the specific case of perturbed normal-form games. The following holds:
    
    \begin{restatable}{proposition}{pertbasis}\label{prop:basis of perturbation}
	    Let $\Gamma=(A_1,A_2,u)$ be a normal-form game, where $A_1=\{a_1,\dots, a_n\}$, and let $p$ be a perturbation for player 1. Let $\Gamma_p=(\tilde{\mathcal{X}}_p, \tilde{\mathcal{Y}}_p, u^*)$ be the generalized normal-form game corresponding to the perturbation (Definition~\ref{def:pert nfg new}). Then the perturbed action space $\tilde{\mathcal{X}}_p$ is a finitely generated convex polytope of dimension $n$, a basis of which is given by the columns of the following invertible matrix:
    	\[
    		B_p \doteq \begin{pmatrix}
    			\tau_p + p(a_1) & p(a_1) & \cdots & p(a_1)\\
    			p(a_2)        & \tau_p + p(a_2) & \cdots & p(a_2)\\
    			\vdots        & \vdots        & \ddots & \vdots\\
    			p(a_n)        & p(a_n) & \cdots & \tau_p + p(a_n)
    		\end{pmatrix}
    	\]
    	where $\tau_p \doteq 1 - p(a_1) - \dots - p(a_n)$.
      \end{restatable}
    
    This means that Algorithm~\ref{algo:gnfg} is applicable and provides a regret-minimizing scheme. We remark that when computing the instantaneous regrets (Algorithm~\ref{algo:gnfg}, Line~12), it is important to remember these values have to be computed against the basis $\{b_1, \dots, b_n\}$ of $\tilde{\mathcal{X}}_p$. However, computing the (expected) utility of the game when player 1 plays according to a mixed strategy is usually more expensive than the same task when player 1 plays a deterministic action. For this reason, we express the instantaneous regret calculation against $\{b_1, \dots, b_n\}$ in terms of instantaneous regrets against the pure actions $\{e_1, \dots, e_n\}$ in the unperturbed game. In particular, by using the fact that the utility function is biaffine, we can write, for each $i\in \{1,\dots,n\}$,
    \begin{align*}
      u(b_i, y_t) &= u\xleft(\tau_p e_i + \sum_{j=1}^{n} p(a_j)e_j, y_t\xright)\\
                  &= \tau_p u(e_i, y_t) + \sum_{j=1}^n p(a_j) u(e_j, y_t),
    \end{align*}
    so that, by introducing $\phi_{t,i} \doteq u(e_i, u_t) - u(x_t, y_t)$ and the corresponding vector $\phi_t \doteq (\phi_{t,1},\dots,\phi_{t,n})^\top$, we have
%    \begin{align*}
%      r_{t,i} &= u(b_i, y_t) - u(x_t, y_t)= \tau_p \phi_{t,i} + \sum_{j=1}^n p(a_j) \phi_{t,j}
%    \end{align*}
    \[
      \vspace{-1mm}
       r_t = r_{t-1} + \tau_p\phi_t + \mathbf{1}\begin{pmatrix}p(a_1)\\ \vdots\\p(a_n)\end{pmatrix}^\top \phi_t.
      \vspace{-1mm}
    \]
	This allows us to compute the regret update in terms of the instantaneous regret against $\{e_1, \dots, e_n\}$ in the unperturbed game, without introducing any overhead from an asymptotic point of view.
	%The final algorithm is presented in Algorithm~\ref{algo:pnfg} (see the Appendix). It is worth noticing that the algorithm is almost completely identical to usual regret-matching for (unperturbed) normal-form games, the only exceptions being Line~5 and Line~10. Furthermore, this variant does not introduce any overhead from an asymptotic point of view.

	The maximum average external regret for player 1 at iteration $T$ is given by Theorem~\ref{thm:max regret polytope}; in this case $|\tilde{\mathcal{X}}_p| = |A_1|$.

%% file: perturbed_efg.tex
\section{Perturbed Extensive-Form Games}

  \begin{algorithm}[htp]
  	\caption{Regret minimization algorithm for perturbed extensive-form games.}
  	\label{algo:cfr}
  	\begin{algorithmic}[1]
  	\Procedure{regret-match$^+$-infoset}{$I, t$}
  	    \Statex $\triangleright$ \textcolor{darkgray}{we assume $A(I) = \{a_1,\dots, a_n\}$.}\vspace{1mm}
		\If{$r^I_{t-1} \in \mathbb{R}_-^n$}
  			\State $x_{t}^I \gets $ any action $\in \tilde{\mathcal{X}}^I_p$
  		\Else
  		    \State $\displaystyle\Lambda^I \gets \sum_{i=1}^n \left[r^I_{t-1}\right]^+_i$
  			\State $\displaystyle x_{t}^I \gets \begin{pmatrix}p(I,a_1)\\\vdots\\p(I,a_n)\end{pmatrix} + \tau_p(I) \frac{\xleft[r^I_{t-1}\xright]^+}{\Lambda^I}$
  		\EndIf
  %		\State $\displaystyle\bar{\sigma}^I \gets \frac{t-1}{t}\bar{\sigma}^I + \frac{1}{t} \sigma^I_t$\vspace{1mm}
  	\EndProcedure
  	\end{algorithmic}
  	\vspace{-1.5mm}
  	\begin{algorithmic}[1]
  	\Procedure{traverse}{$h, i, t, \pi_1, \pi_2$}
  	  \Statex $\triangleright$ \textcolor{darkgray}{assume $h$ belongs to information set $I$}
  	  \If{$h\in Z$}
  	    \State\Return $u(h)$
  	  \EndIf{}
  	    \If{$P(h)=c$}\Comment{\textcolor{darkgray}{chance node}}
  	      \State sample $a\sim f_c(h)$
  	      \State\Return \Call{traverse}{$ha, t, \pi_1, \pi_2$}
  	    \ElsIf{$P(h) = 2$}\vspace{1mm}
  	      \State $\displaystyle v^I_t \gets \begin{pmatrix}
  	        \textsc{traverse}{(ha_1, t, \pi_1,y^I_{t,1}\pi_2)}\\\vdots\\\textsc{traverse}{(ha_n, t, \pi_1,y^I_{t,n}\pi_2)}
  	      \end{pmatrix}$\vspace{1mm}
  	    \Else{}\Comment{\textcolor{darkgray}{player 1's turn}}\vspace{1mm}
  	      \State\Call{regret-match-infoset}{$I, t$} \vspace{1mm}
  	      \State $\displaystyle v^I_t \gets \begin{pmatrix}
  	              \textsc{traverse}{(ha_1, t, x^I_{t,1}\pi_1, \pi_2)}\\\vdots\\\textsc{traverse}{(ha_n, t, x^I_{t,n}\pi_1,\pi_2)}
  	    	      \end{pmatrix}$\vspace{1mm}
  		   \State $\bar v \gets (x^I_t)^\top v^I_{t}$\vspace{1mm}
  		   \State $\phi^I_{t} \gets \pi_{2}(v^I_{t} - \textbf{1}\, \bar v)$\vspace{1mm}
  		   \State $\displaystyle r^I_{t} \gets\hspace{-.5mm}\left[r^I_{t-1} + \tau_p(I)\phi^I_{t} + \mathbf{1}\hspace{-1mm}\begin{pmatrix}p(I,a_1)\\ \vdots\\p(I,a_n)\end{pmatrix}^{\hspace{-2mm}\top} \phi^I_{t}\right]^+$\vspace{1mm}
  		   \State $\bar x^I \gets \bar x^I + \pi_1 x_t^I$
  		 \EndIf
  	  \State\Return $\bar v$
  	\EndProcedure
  	\end{algorithmic}
  	\vspace{-1.5mm}
  	\begin{algorithmic}[1]
  	\Procedure{CFR$^+$}{$\Gamma$}\Comment{\textcolor{darkgray}{\small $\Gamma=(H, Z, A, P, f_c, \mathcal{I}_1, \mathcal{I}_2, u)$}}
  		\For{\textbf{all} $I\in\mathcal{I}_1$}
  			\State $r_0^I \gets (0,\dots, 0)^\top \in \mathbb{R}^{|A(I)|}$\vspace{1mm}
  			\State $\bar x^I \gets (0,\dots, 0)^\top \in \mathbb{R}^{|A(I)|}$
  		\EndFor
  		\For{$t = 1, 2, 3, \dots$}
  	      \State play according to strategy $x_t$
  	      \State observe strategy $y_t$ played by opponent
  	      \State\Call{traverse}{$\varnothing,t,1,1$}
  		\EndFor
  		\For{\textbf{all} $I\in\mathcal{I}_1$}
  		  \State $\bar x^I \gets \left({\bar x^I}/{\sum_{i=1}^{|A(I)|}\bar x^I_i}\right)$
  		\EndFor
  		\Statex $\triangleright$ \textcolor{darkgray}{$\bar x$ contains the average strategy for player 1.}
  	\EndProcedure
  	\end{algorithmic}
  \end{algorithm}

  As discussed in Section~\ref{sec:behavioral constraints}, it is possible to
  use CFR in conjunction with any regret-minimizing scheme for generalized NFGs,
  in order to define a regret-minimizing scheme able to support any
  behaviorally-perturbed EFG (thus including the restricted case of perturbed
  EFGs). In Algorithm~\ref{algo:cfr}, we propose an implementation of CFR$^+$,
  i.e. CFR instantiated with the RM$^+$ algorithm
  able to handle perturbed EFGs. Algorithm~\ref{algo:cfr} assumes that we are
  given a perturbation $p$ of the extensive-form game, $\cX_p^I$ denotes the
  perturbed simplex for information set $I$, and $\tau_p(I)$ is as in
  Proposition~\ref{prop:basis of perturbation}.
  
  The following theorem characterizes the convergence guarantee of the proposed algorithm.	
  \begin{restatable}{theorem}{cfrconvpert}\label{thm:regret cfr pert}
	Let $(H, Z, A, P, f_c, \mathcal{I}_1, \mathcal{I}_2, u)$ be an exten\-sive-form game; let $p$ be the perturbation applied to the game. The maximum average external regret for player 1 in the perturbed game at iteration $T$, when player 1 plays according to Algorithm~\ref{algo:cfr}, is bounded by
    \[
      \vspace{-1mm}
      \bar R^T_1 \le \gamma |\mathcal{I}_1|\frac{\sqrt{\max_{I\in\mathcal{I}_1} |A(I)|}}{\sqrt{T}},
      \vspace{-1mm}
    \]
    where $\gamma \doteq \max_{x, y} u(x, y) - \min_{x,y} u(x,y)$.
  \end{restatable}
  \begin{proof}
    Follows as a corollary of Theorem~\ref{thm:regret cfr}.
  \end{proof}
  
  Notice that the bound provided by Theorem~\ref{thm:regret cfr} is the same provided by the original CFR algorithm proposed by \citet{Zinkevich07:Regret}. In other words, our modification does not impair the convergence and speed guarantees given by the original algorithm.

%% file: experiments.tex
\section{Experimental Evaluation}
\newcommand{\shortcite}{\cite}\vspace{-2mm}

We conducted experiments to investigate the practical performance of our
perturbed-regret-minimization approach when used to instantiate the CFR and
CFR+ algorithms for computing approximate EFPE in EFGs.
We compare these algorithms to state-of-the-art Nash-equilibrium-finding
algorithms: EGT~\cite{Nesterov05:Excessive} on an unperturbed polytope using the state-of-the-art smoothing technique by
Kroer et. al.~\shortcite{Kroer17:Theoretical}, CFR~\cite{Zinkevich07:Regret} and
CFR+~\cite{Tammelin15:Solving}.
We conducted the experiments on Leduc hold'em poker~\cite{Southey05:Bayes}, a
widely-used benchmark in the imperfect-information game-solving community. In
our variant, \emph{Leduc k}, the deck consists of $k$ pairs of cards $1\ldots
k$, for a total deck size of $2k$. We experiment on the standard Leduc game where $k=3$ and a larger game where $k=5$.
Each player initially pays one chip to the
pot, and is dealt a single private card. After a round of betting, a community
card is dealt face up. After a subsequent round of betting, if neither player
has folded, both players reveal their private cards. If either player pairs
their card with the community card they win the pot. Otherwise, the player with
the highest private card wins. In the event that both players have the same
private card, they draw and split the pot. We consider $k\in \left\{ 3,5
\right\}$. We test our approach on games subject to different uniform perturbations $p(I, a) = \xi$ for all information sets $I$ and actions $a \in A(I)$, for $\e \in \left\{ 0.1,0.05, 0.01, 0.005, 0.001 \right\}$.

Figure~\ref{fig:experiments_leduc5_eps} reports on convergence to Nash
equilibrium. The x-axis shows the number of tree traversals performed. We use
tree traversals rather than iterations because EGT requires more tree traversals
than CFR+ per iteration. The y-axis shows the sum of player regrets in the full
(unperturbed) game. For both Leduc 3 and 5, we find that the $\e$ perturbations
have only a small effect on overall convergence rate until convergence within
the perturbed polytope, at which point the regret in the unperturbed game stops
decreasing, as expected. Until bottoming out, the convergence is almost
identical for all CFR+ algorithms. This shows that our approach can be utilized
in practice: there is no substantial loss of convergence rate. Later in the run
once the perturbed algorithms have bottomed out, there is a tradeoff between
exploitability in the full game and refinement (i.e., better performance in
low-probability information sets).
\begin{figure}[t]
  \centering
  \includegraphics[scale=.68]{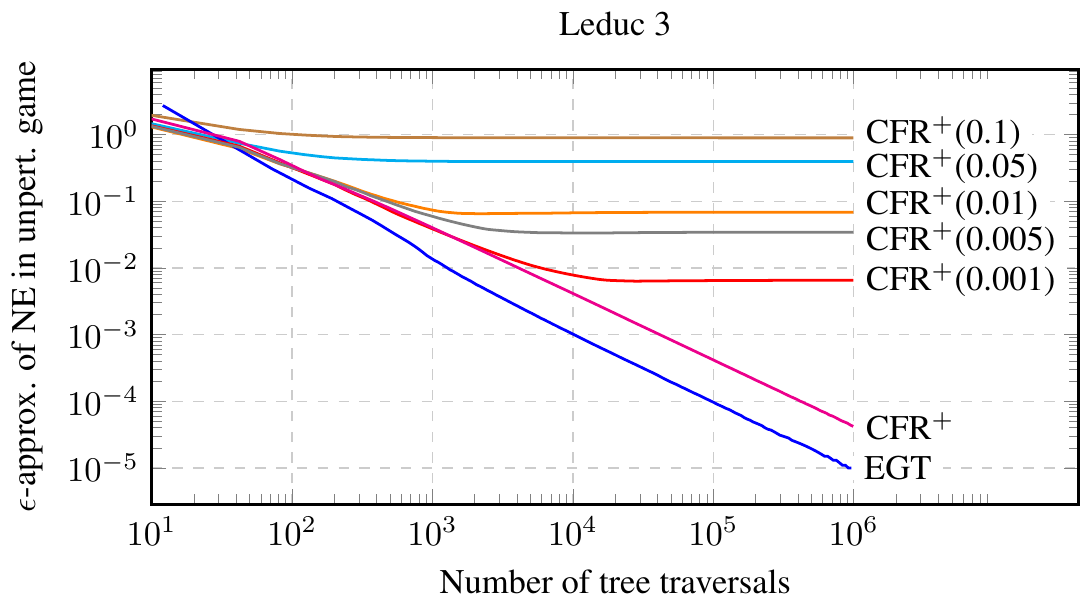}
  \includegraphics[scale=.68]{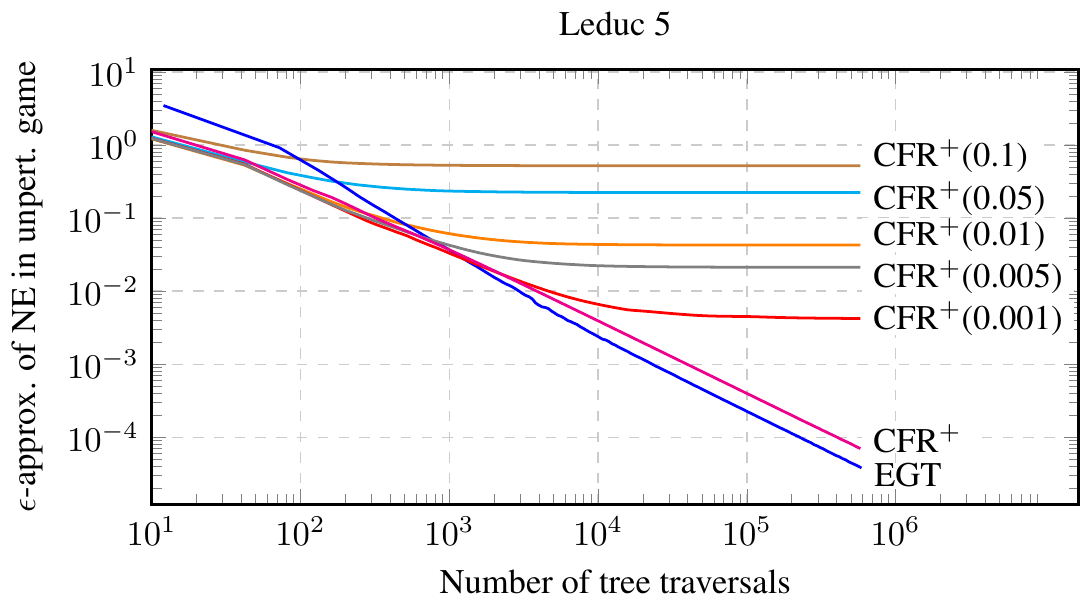}
  % \includegraphics[width=0.49\textwidth]{plots/pgf/leduc5_eps_unperturbed.pdf}

  % \begin{tikzpicture}[scale=.80]
  %   \input{../plots/regr_unpert_leduc5}
  % \end{tikzpicture}
  %
  % \includegraphics[width=1\linewidth]{aux/eps_approx_notitle.pdf}
  \vspace{-.1in}
  \caption{Regret in Leduc 3 and Leduc 5 as a function of the number of
    iterations for {\egt} and CFR+ with various $\e$ perturbations (denoted in
    parentheses). Both axes are on a log scale.}
  \label{fig:experiments_leduc5_eps}
  \vspace{-5mm}
\end{figure}

The second set of experiments, Figure~\ref{fig:experiments_leduc5_max_regret}, investigates the improvement that our perturbation
approach achieves compared to standard Nash equilibrium solutions in terms of
equilibrium refinement. Our measure of refinement is the maximum
regret at any information set, conditioned on reaching that information set. As discussed, convergence to a Nash Equilibrium does not guarantee that this measure goes to zero. Again, the
x-axis shows the number of tree traversals performed. The y-axis
shows the maximum regret at any individual information set. Both unperturbed
CFR+ and EGT perform badly in both games with respect to this measure of
refinement. In Leduc 5, both have maximum regret two orders of magnitude worse
than the perturbed approach. In Leduc 3, EGT still
does as poorly. CFR+ does slightly better, but is still worse than our stronger
refinements by more than an order of magnitude.
The maximum regret one can possibly cause in an
information set in either Leduc game is 23,
so CFR+ and unperturbed EGT also do poorly in that sense.

In contrast to this, we find that our $\e$-perturbed
solution concepts converge to a strategy with low regret at every information
set. The choice of $\e$ is important: for $\e=0.001$, the smallest perturbation,
we see that it takes a long time to converge at low-probability information
sets, whereas we converge reasonably quickly for $\e=0.01$ or $\e=0.005$; for
$\e=0.1$ and $\e=0.05$ the perturbations are too large, and we end up converging
with relatively high regret (due to being forced to play every action with
probability $\e$). Thus, within this set of experiments, $\e\in
\left[0.005,0.01\right]$ seems to be the ideal amount of perturbation.

\begin{figure}[t]
  \centering
  \includegraphics[scale=.68]{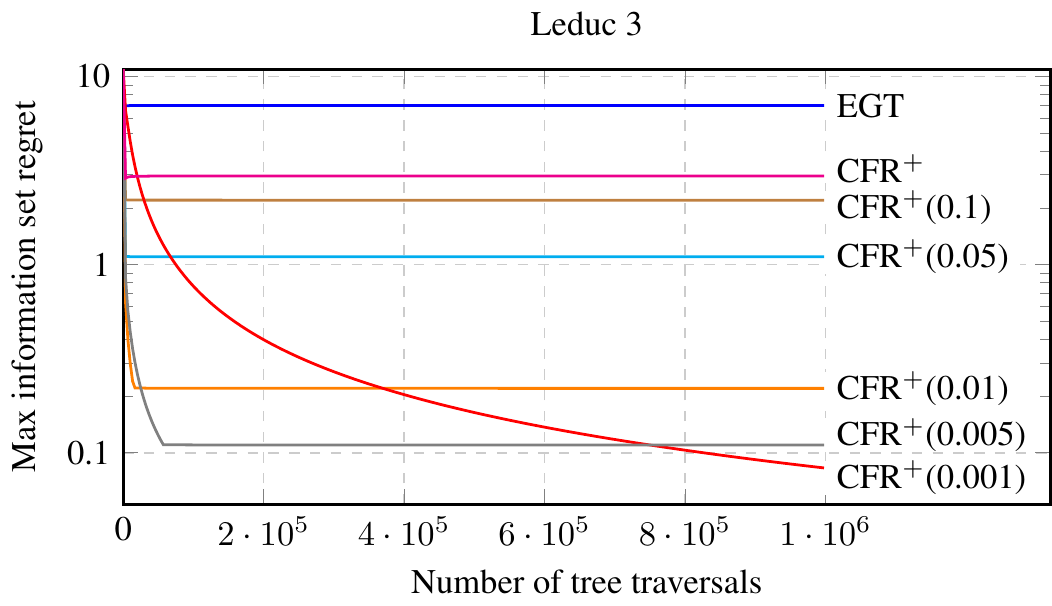}
  \includegraphics[scale=.68]{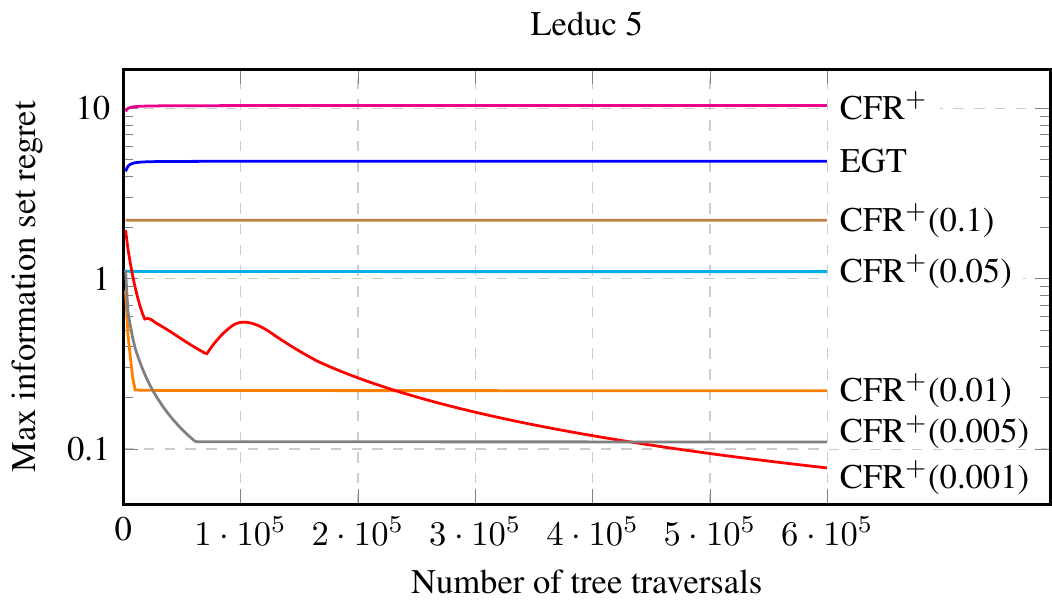}
  % \includegraphics[width=0.49\textwidth]{plots/pgf/leduc5_max_infoset_regret.pdf}

  % \begin{tikzpicture}[scale=.80]
  %  \input{../plots/max_regr_leduc5}
  %\end{tikzpicture}

	% \includegraphics[width=.98\linewidth]{aux/max_regr_notitle.pdf}
  \vspace{-.1in}
  \caption{Maximum regret at any individual information set in Leduc 3 and
    Leduc 5, as a function of the number of iterations, for standard {\egt} as
    well as with various $\e$ perturbations (denoted EGT($\e$)) and CFR+. The y-axis is on a log scale.}
  \label{fig:experiments_leduc5_max_regret}
  \vspace{-5mm}
\end{figure}

%% file: conclusion.tex
\section{Discussion}

We extended the RM and {\rmp} regret minimization algorithms to
finitely-generated convex polytopes, and specialized our results to
linearly-constrained simplexes and behaviorally-perturbed EFGs. We then showed
how this allows us to compute an approximate EFPE. Our experiments showed that
this approach leads to much stronger strategies for information sets reached
with low probability, while maintaining the strong convergence rate of CFR$^+$.

Our experiments raise an interesting question. Across both games, we see that
the maximum information set regret goes down much faster for larger amounts of
perturbation, but then it bottoms out earlier than for smaller perturbations (as
expected). To get the best of both large and small perturbations, it may be
possible to decrease the perturbations over time, leading to faster convergence
rate, while never bottoming out. However, this has a number of challenges
associated with it. Most importantly, we need a variant of RM or {\rmp} that can
handle a slowly expanding feasible set within the simplex. This would also
require decreasing the perturbations at the correct rate; if decreased too
quickly, it is unlikely that we will converge to a refinement, and if decreased
too slowly, we might still bottom out.

% From a theoretical perspective, one could set the perturbations in our approach
% according to those needed in \citet{Farina17:Extensive-Form}. In that case,
% exact solutions to the perturbed game can be used to extract an exact EFPE.
% \citet{Mehrotra91:Finding} shows that an exact solution to  can be recovered from an
% approximate

The CFR algorithms have been shown to work well with a number of other
techniques, notably sampling~\citep{Lanctot09:Monte} and
abstraction~\citep{Lanctot12:no-regret,Kroer16:Extensive_Imperfect}. It would be
both theoretically and practically interesting to see how well our refinement
approach works in conjunction with these techniques.

%% file: appendix.tex
\section*{Appendix A: Omitted Proofs}

\pertbasis*
\begin{proof}
        ($\subseteq$) Call $b_1, \dots, b_n$ the columns of $B_p$. Given $\lambda_1, \dots, \lambda_n \in \mathbb{R}_+$ with $\lambda_1 + \dots + \lambda_n = 1$, we have
        \begin{align*}
          \sum_{i=1}^n\lambda_i b_i = \sum_{i=1}^n (p(a_i) + \tau_p \lambda_i) e_i \in \tilde{\mathcal{X}_p}.
        \end{align*}
        This proves that the convex hull of $\{b_1,\dots, b_n\}$ is a subset of $\tilde{\mathcal{X}}_p$.
        
        ($\supseteq$) Let $x \in \tilde{\mathcal{X}}_p$. There exist $\lambda_1, \dots, \lambda_n$, with $\lambda_i \ge p(a_i)$ for all $i\in\{1,\dots,n\}$ and $\lambda_1+\dots+\lambda_n=1$, such that
        \[
          x = \lambda_1 e_1 + \dots + \lambda_n e_n.
        \]
        Define
        \[
          \mu_i \doteq \frac{\lambda_i - p(a_i)}{\tau_p} \ge 0, \qquad i\in\{1,\dots, n\}.
        \]
         By using the fact that $\sum_i \lambda_i = 1$ and the definition of $\tau_p$, we find
         \[
           \mu_1 + \dots + \mu_n = \frac{\sum_{i}\lambda_i - \sum_i p(a_i)}{\tau_p} = 1.
         \]
         We have
		 \begin{align*}
		   \sum_{i=1}^{n} \mu_i b_i &= \sum_{i=1}^n \frac{\lambda_i - p(a_i)}{\tau_p} b_i\\
		                            &= \sum_{i=1}^n \frac{\lambda_i - p(a_i)}{\tau_p} \left(\tau_p e_i + \sum_{j=1}^n p(a_j) e_j \right)\\
		                            &= \xleft(\sum_{i=1}^n \mu_i\xright)\sum_{j=1}^n p(a_j)e_j+\sum_{i=1}^n \lambda_i e_i - p(a_i)e_i\\
		                            &= x + \sum_{j=1}^n p(a_j) e_j - \sum_{i=1}^n p(a_i)e_i
		                            = x.
		 \end{align*}
        This proves that the convex hull of $\{b_1, \dots, b_n\}$ is a superset of $\tilde{\mathcal{X}}_p$, completing the proof.
      \end{proof}
      
\cfrconv*
  \begin{proof}
    Consider the proof of Theorem~3 given by~\citet[Appendix
    A.1]{Zinkevich07:Regret}. In order to account for the constrained action
    spaces $\mathcal{Q}^I$, the definition of (average) full counterfactual
    regret for player $i$ up to iteration $T$ relative to information set $I$
    has to be changed as follows (we use the same symbols and notation as in
    their paper):
    \[
      \bar R_{i,\text{full}}^T = \frac{1}{T} \max_{\sigma' \in \tilde\Sigma_i} \sum_{t=1}^T \pi_{-i}^{\sigma^t} \mathbb{E}(u_i(\sigma^t \mid_{D(I)\to \sigma'}) - u_i(\sigma^t,I)).
    \]
    In other words, we substitute the domain of the max, originally $\Sigma_i$, substituting the set of unconstrained strategies over the information sets in $D(I)$, with the space $\tilde{\Sigma}_i$ of strategies compatible with the constrained action polytopes $\mathcal{Q}^{I'}$, for all $I'\in D(I)$. The same substitution propagates in Equations (10) -- (13).
    
    Similarly, in Equations (10) -- (14) we substitute the max operation $\max_{a \in A(I)}$ with $\max_{q \in \mathcal{B}^I}$, where $\mathcal{B}^I$ is any fixed convex basis of choice for $\mathcal{Q}^I$.
    This modification needs to be applied to the definition on average immediate counterfactual regret (Equation~6), so that
    \[
      \bar R^T_{i,\text{imm}} = \frac{1}{T} \max_{q\in\mathcal{B}^I} \sum_{i=1}^T \pi_{-i}^{\sigma^t}(I)\mathbb{E}(u_i(\sigma^t \mid_{I\to q}, I)-u_i(\sigma^t, I))
    \]
    Intuitively, this is consistent with the crucial point in the construction of the proof of Theorem~\ref{thm:generalized hannan} of this paper.
    Furthermore, this makes it so that the first part of the expression on the right-hand side of Equation~(11) is the (average) immediate regret against the basis vector $q$ of $\mathcal{Q}^I$, 
    
    This proves that CFR defines a regret-minimizing scheme on behaviorally-constrained EFGs. Specifically, for each player $i\in\{1, 2\}$,
    \[
      \bar R^T_i \le \sum_{I \in \mathcal{I}_i} \bar R^{T,+}_{i,\text{imm}}(I).
    \]
    We can then locally minimize the immediate regret at each information set with a regret-minimizing scheme working on $\mathcal{Q}^I$, in this case RM$^+$ (Algorithm~\ref{algo:gnfg}). The same proof of Theorem~4 from~\citet{Zinkevich07:Regret} can be applied, substituting the convergence bound of the unconstrained version of RM$^+$ with the new, constrained version (Theorem~\ref{thm:max regret polytope} of the present paper).
  \end{proof}
  
\section*{Appendix B: Omitted Pseudocode}
	\begin{algorithm}[H]
		\caption{Regret minimization algorithm for perturbed normal-form games.}
		\label{algo:pnfg}
		\begin{algorithmic}[1]
		\Procedure{perturbed-regret-matching$^+$}{$\Gamma, p$}
		    \Statex $\triangleright$ \textcolor{darkgray}{\textbf{note}: this reflects the point of view of player 1}
%			\Statex $\triangleright$ \textcolor{darkgray}{$c_i$ is the $i$-th column of $B_p$.}
		    \State $r_0 \gets (0, \dots, 0)^\top \in \mathbb{R}^n$
		    \State $\bar{x} \gets (0, \dots, 0)^\top \in \mathbb{R}^n$
			\For{$t = 1, 2, 3, \dots$}
			  \If{$r_{t-1} \in \mathbb{R}_-^n$}
			    \State $x_{t} \gets $ any action $\in \tilde{\mathcal{X}}_p$
			  \Else{}
			    \State $\Lambda_{t-1} \gets \displaystyle\sum_{i=1}^n\, [r_{t-1}]^+_i$
			    \State $x_t \gets \displaystyle \begin{pmatrix}
			    	    p(a_1)\\
			    	    \vdots\\
			    	    p(a_n)
			          \end{pmatrix}+ \tau_p \frac{[r_{t-1}]^+}{\Lambda_{t-1}}$
			  \EndIf{}
			  \State play action $x_t$
			  \State observe action $y_t \in \mathcal{Y}$ played by opponent\vspace{1mm}
			  \State $\displaystyle \phi_{t} \gets \begin{pmatrix}
			  				  u(e_1, y_t) - u(x_t, y_t)\\\vdots\\
			  				  u(e_n, y_t) - u(x_t, y_t)
			  			  \end{pmatrix}$\vspace{1mm}
			  \State $r_{t} \gets \left[r_{t-1} + \tau_p\phi_{t} + \mathbf{1}\begin{pmatrix}p(a_1)\\ \vdots\\p(a_n)\end{pmatrix}^\top \phi_{t}\right]^+$\vspace{1mm}
			  \State $\displaystyle \bar{x} \gets \frac{t-1}{t}\bar x + \frac{1}{t}x_t$\vspace{1mm}
			\EndFor{}
			\Statex $\triangleright$ \textcolor{darkgray}{$\bar x$ contains the average strategy for player 1}
		\EndProcedure
		\end{algorithmic}
	\end{algorithm}  